\documentclass{article}
\usepackage{amssymb}
\usepackage{amsthm}
\usepackage{amsmath}
\usepackage[dvips]{epsfig}
\usepackage{graphicx}
\usepackage{color}
\usepackage{url}

\newtheorem{theorem}{\bf Theorem}[section]

\newtheorem{proposition}{\bf Proposition}[section]

\newcommand{\bw}{{\bf w}}

\newcommand{\bq}{{\bf q}}
\newcommand{\PT}{{\cal PT}}
\newcommand{\T}{{\cal T}}
\newcommand{\p}{{\cal P}}
\newcommand{\tbw}{\tilde{\bf w}}
\newcommand{\tE}{\tilde{E}}

\newcommand{\hG}{\hat{G}}
\newcommand{\hJ}{\hat{J}}
\newcommand{\hE}{\hat{E}}

\begin{document}

\title{Nonlinear modes in a~generalized $\PT$-symmetric discrete nonlinear Schr\"odinger equation}

\author{ D. E. Pelinovsky\\ Department of Mathematics and Statistics, McMaster
University,\\ Hamilton, Ontario, Canada, L8S 4K1\\[5mm]
D. A. Zezyulin and V.~V.~Konotop\\ Centro de F\'isica Te\'orica e Computacional  and Departamento de F\'isica,\\
Faculdade de Ci\^encias, Universidade de Lisboa, \\Avenida Professor Gama Pinto 2, Lisboa 1649-003, Portugal
}

\maketitle

\begin{abstract}
We generalize a finite parity-time ($\PT$-)  symmetric network of the discrete nonlinear Schr\"{o}dinger type
and obtain general results on linear stability of the zero equilibrium,
on the nonlinear dynamics of the dimer model, as well as on the
existence and stability of large-amplitude stationary nonlinear modes.
A result of particular importance and novelty is the classification of
all possible stationary modes in the limit of large amplitudes.
We also discover a new integrable configuration of a $\PT$-symmetric dimer.
\end{abstract}

\section{Introduction}
\label{intro}

In the present work we consider a generalized finite
network of the discrete nonlinear Schr\"{o}dinger (dNLS) type
with gain and dissipation terms:
\begin{align}
\label{dnls-general}
i \frac{d q_n}{d t} = q_{n+1} + q_{n-1} + i \gamma_n q_n + \left[(1-\chi_n)|q_{n}|^2+\chi_n|q_{1-n}|^2\right] q_n.
\end{align}
Here $t$ is the evolution variable, $n$ is the integer site number between $-N+1$ and $N$,
and the real-valued coefficients satisfy
\begin{align}
\label{parameters}
\gamma_n=-\gamma_{1-n}, \quad \chi_n=\chi_{1-n}.
\end{align}
The network is closed with the boundary conditions
$q_{-N}=q_{N+1}=0$, which correspond to an isolated array of $2N$
nonlinear oscillators. The array is   $\p \T$ symmetric
with respect to the parity transformation $\p$ about the center point in $[-N+1,N]$
and the time reversion $\T$. Introducing the column-vector $\bq=(q_{-N+1},...,q_N)^T$
(hereafter $T$ stands for the matrix transpose), we rewrite
Eq.~(\ref{dnls-general})  in the vectorial form
\begin{align}
\label{dnls-general-vec}
i \frac{d \bq}{d t} = \left[ H + i G + F(\bq) \right] \bq,
\end{align}
where the real valued $2N\times 2N$ matrices are given by their entries as follows:
\begin{align}
\label{HG}
H_{nm}=\delta_{n,m+1}+\delta_{n,m-1},\quad  G_{nm}=\gamma_n\delta_{n,m},
\end{align}
and
\begin{align}
\label{F}
\left[F(\bq)\right]_{n,m}= \left[(1-\chi_n)|q_{n}|^2+\chi_n|q_{1-n}|^2\right]\delta_{n,m},
\end{align}
with $\delta_{n,m}$ being the Kronecker symbol and $n$ running from $-N+1$ to $N$.

The $\p$ operator is defined as a $2N\times2N$ antidiagonal matrix
with the only nonzero entries $\p_{kj}=\delta_{j,2N+1-k}$ (with $k$ and $j$ running from $1$ to $2N$). The
$\T$ operator is defined as complex conjugation. These definitions ensure
\begin{align}
\label{comut}
[\p,H]=\{\p,G\}=0, \quad [\PT,H+iG]=0,
\end{align}
where $[\cdot,\cdot]$ and $\{\cdot,\cdot\}$ stand for commutator and anti-commutator, respectively.
The conditions (\ref{comut}) formalize the definition of the $\PT$-symmetric linear lattice.
If all the eigenvalues of the operator $H+iG$ are real then the linear system is said to belong to the
unbroken $\PT$-symmetric phase (this situation obviously corresponds to  linear stability of
the zero equilibrium of the full nonlinear problem). Phases of  broken   $\PT$ symmetry arise
when the linear operator   has a pair (or several pairs) of complex conjugated eigenvalues.

Motivation for the  study  of nonlinear system (\ref{dnls-general}) stems from the recently
growing interest in nonlinear dynamics of the $\PT$-symmetric networks of the dNLS type which can be viewed as particular limits of
the network (\ref{dnls-general}). The simplest case of $N=1$ with $\chi_1 = 0$ corresponds
to the nonlinear $\PT$-symmetric dimer with the on-site (Kerr-type) nonlinearity. The latter model
was shown~\cite{Ramezani} to be a fully
integrable system, and its dynamics was thoroughly investigated in~\cite{Sukhorukov,Pelin4,Susanto,Barashenkov}.
Setting $N=2$  one  arrives   at a nonlinear $\PT$-symmetric quadrimer, which has been considered
for $\gamma_1 = \gamma_2$ in~\cite{LK, ZK},  as well as   for $\gamma_1 \neq \gamma_2$ in \cite{ZK, ZK2}.
A peculiarity of the latter case consists of the existence of multiple points (like triple points)
corresponding to the system parameters at which more than two (e.g. one unbroken and
two broken $\PT$-symmetric) phases of the zero equilibrium co-exist.

In a more general case of a network of arbitrary finite size  ($N < \infty$) the linear stability of the zero equilibrium was investigated in \cite{fragil,Lin09,Joglekar10,LKD,Pelin3,Bar} for the following particular configurations:
\begin{enumerate}
\item a chain with a $\PT$-symmetric \textit{defect}, i.e. having two ``defect'' nodes
$\gamma_{d}=-\gamma_{1-d} = \gamma$ at some $d=1, 2, \ldots, N$ (with all other sites having no dissipation or gain);

\item an \textit{alternating} chain with $\gamma_n = \gamma(-1)^n$ (where the sites with equal dissipation and gain alternate);

\item a \textit{clustered} chain with  $\gamma_n=\gamma {\rm sign}\left(n-\frac{1}{2}\right)$ (thus having two intervals of sites: one with gain and another one with dissipation).
\end{enumerate}
Stability of the zero equilibrium in unbounded $\p\T$-symmetric
dNLS chains (both alternating and clustered) was studied in \cite{Dmitr10,Pelin2}.

Stationary nonlinear modes  can be represented as $\bq(t)=\bw e^{-iEt}$, where $\bw=(w_{1-N},...,w_N)^T$
is the time-independent column-vector, satisfying the system of the nonlinear equations
\begin{align}
\label{dnls-stat-general}
Ew_n = w_{n+1} + w_{n-1} + i \gamma_n w_n + \left[(1-\chi_n)|w_{n}|^2+\chi_n|w_{1-n}|^2\right]  w_n.
\end{align}
The spectral parameter $E$ will be termed energy (in optical applications it corresponds to the propagation constant).
Focusing on the $\PT$-invariant stationary modes, i.e. the modes satisfying $\PT \bw = \bw$
(see \cite{Yang} for the discussion of relevance of this requirement), we have the reduction
$w_n=\bar{w}_{1-n}$ (hereafter an overbar stands for the complex conjugation), which reduces   Eq.~(\ref{dnls-stat-general})
 to the system of $N$ algebraic equations
\begin{align}
\label{stationary-general}
E w_n = w_{n+1} + w_{n-1} + i \gamma_n w_n + |w_{n}|^2w_n, \quad 1 \leq n \leq N,
\end{align}
subject to the boundary condition $w_0 = \bar{w}_1$ and $w_{N+1} = 0$. We note that the
parameters $\chi_n$ are eliminated in the stationary system  (\ref{stationary-general}).
The  $\PT$-invariant nonlinear modes obeying   system  (\ref{stationary-general}) have been found to bifurcate from the  eigenstates of the underlying linear operator $H+iG$~\cite{ZK, Pelin3}. Nonlinear modes bifurcating from exceptional points of the underlying linear operator
  were considered in~\cite{ZK2}.

If all $\gamma_n$ and $\chi_n$ are zero,  the nonlinear system (\ref{dnls-general}) is reduced to a chain of coupled conservative
oscillators with on-site nonlinearity. Triggered by the work~\cite{Eilbeck},
this model (known as the dNLS equation) has been extensively studied during
the last thirty years~\cite{Surv_Tsir1,Pel-book}. A powerful analytical tool in this study, introduced in~\cite{Eilbeck}, is the analytical continuation of the localized modes from the anticontinuum limit, i.e. from the limit $E\to\infty$, allowed to prove the existence of discrete breathers~\cite{McKay1994}. The approach allows for
classification of the nonlinear modes~\cite{Eilbeck,ELS85,Alfimov} and systematic study of their stability~\cite{CE85,pkf,pelsak}. Recently, it was shown in~\cite{KPZ}, that the approach based on the anticontinuum limit can be extended to   $\PT$-symmetric networks, and in particular, to the alternating chain with $\gamma_n=(-1)^n\gamma$. Further studies of this model were performed in~\cite{Pelin3} where the nonlinear modes were described both for an isolated $\p \T$-symmetric network and for an embedded $\p \T$-symmetric chain as a defect in an infinite dNLS equation.

In the present work, we elaborate the anticontinuum limit, which correspond to the limit
of large energies,  $E \to \infty$,  for the generalized $\PT$-symmetric dNLS network (\ref{dnls-general}) and offer a complete
classification of stationary nonlinear modes in the system of algebraic equations (\ref{stationary-general}). We prove that
large-amplitude stationary modes exist only if the large-amplitude
sites are all grouped together near the center point in $[-N+1,N]$ and no other stationary modes
exist. We classify the stationary modes according to the binary roots
of the phase equations and consider stability of the corresponding configurations.
The outcome of stability computations is similar to the stability theorem
in \cite{pkf} but fewer configurations are stable in the $\p \T$-symmetric case
because of the amplitude growth of the small-amplitude sites with gains.

Another goal of our study is to make a step  towards description of
the dynamical properties of the nonlinear system   (\ref{dnls-general}).
We consider the dynamics of the dimer model ($N=1$) and, in particular,
show that for $\chi_1= \frac{1}{2}$ all  the time-dependent solutions
are bounded for sufficiently small values of $\gamma_1$ that ensure that
the $\PT$ symmetry of the underlying linear problem is unbroken.
We also show that the model with $\chi_1=  \frac{1}{2}$ admits two integrals
of motions and hence is an integrable model, similarly to its counterpart
with $\chi_1=0$ considered in \cite{Ramezani}. On the other hand,  we show that
for any $\chi_1 \neq \frac{1}{2}$ the dimer model always has unbounded solutions
for any arbitrarily small but nonzero $\gamma_1$.

The paper is organized as follows. Section~\ref{sec:lin} gives sufficient conditions
for existence of unbroken and broken $\PT$-symmetric phases obtained from the linear
stability analysis of the zero equilibrium. Section~\ref{sec:dim} characterizes the
nonlinear dynamics of the dimer model for $N = 1$. Conserved quantities of the generalized
dimer model with $\chi_1 = \frac{1}{2}$ are also discussed. Section~\ref{sec:class}
describes the existence and classification of the
stationary  nonlinear modes in the   case of arbitrary finite $N$. Section~\ref{sec:stab}
outlines the stability
computations for the most important configurations of the stationary  nonlinear modes.
Section~\ref{sec:concl} concludes
the paper with a summary of results.

\section{Linear stability of the zero equilibrium}
\label{sec:lin}
While the present work is devoted to the  nonlinear problem (\ref{dnls-general}), 
the underlying  linear model itself makes a particular physical meaning
if it describes stable propagation of linear waves.
Such waves are obtained from the linear eigenvalue problem
\begin{align}
\label{lin}
\tE\tbw = J\tbw, \quad J:=H + i G
\end{align}
with $H$ and $G$ being the matrices defined by (\ref{HG}). Hereafter we use tildes in order to distinguish eigenvectors $\tbw$ and eigenvalues $\tE$ of the underlying linear problem.
In this section  we formulate the sufficient conditions for existence of the
propagating linear modes and recover some of the known relevant results
of the linear theory.

For $G=0$, matrix $J$ becomes Hermitian and its spectrum is real.
One can also expect that  if all $\gamma_n$ are sufficiently
small, then the spectrum of $J$ remains real. In this situation,
$J$ is said to have \textit{unbroken} $\PT$
symmetry. On the other hand, if $\gamma_n$ are large enough, then
the spectrum enters the complex domain, i.e. $\PT$ symmetry  becomes
\textit{broken}.  In  following Theorems~\ref{exact} and
\ref{breaking} we substantiate the above informal discussion by
finding (possibly not optimal) estimates for domains  of unbroken and
broken $\PT$ symmetries for arbitrary finite $N$.
Before passing to Theorems~\ref{exact} and \ref{breaking}, let us
prove the following auxiliary result.

\begin{proposition}
\label{prop:char}%
The characteristic polynomial of the matrix $J$
$$
P(\lambda) = \lambda^{2N} - p_1\lambda^{2N-1}  -  p_2 \lambda^{2N-2} - \ldots - p_{2N}
$$
has real coefficients $p_n$. All the
coefficients associated with odd powers of $\lambda$ are zero:
$$
p_1 = p_3 = \ldots = p_{2N-1}=0.
$$
Additionally one has
\begin{align}
\label{eq:p2}%
 p_{2} = 2N -1 - \sum_{n=1}^N \gamma_n^2.
\end{align}
\end{proposition}

\begin{proof}
Reality of coefficients of the characteristic polynomial $P(\lambda)$  for a
general $\PT$-symmetric matrix   follows from the fact that any
eigenvalue of the $\PT$-symmetric matrix is either real or belongs to a complex-conjugate pair
\cite{BenderNessSuff}. In order to show that the characteristic
polynomial $P(\lambda)$ does not contain odd powers of $\lambda$, we notice
that matrices $J$ and $-J$ are related by the similarity transformation
$$
-J=Z^{-1}JZ,
$$
where $Z$ is a matrix with only nonzero entries
$Z_{nj}= (-1)^{j-1} \delta_{j,2N+1-n}$ for $j,n=1, 2, \ldots, 2N$.
Notice that $\det Z = 1$  and $Z^{-1} = -Z$.  Therefore, matrices
$J$ and $-J$ share the same characteristic polynomial
$P(\lambda)$. The latter implies that $P(\lambda)$ is an even
function.

In order to prove (\ref{eq:p2}), we first notice that matrix $J$
is traceless: $\textrm{tr}\, J=0$. In this case,
the coefficient $p_2$ is given by $p_2=\frac{1}{2}\textrm{tr}\, J^2$
\cite{Gantmakher}. Using a simple straightforward computation,  one
finds that $\textrm{tr}\, J^2 = 2\left(2N -1-\sum_{n=1}^N
\gamma_n^2 \right)$, which yields (\ref{eq:p2}).
\end{proof}

\begin{theorem}[on unbroken $\PT$ symmetry]
\label{exact} %
Define $\Gamma := \max\limits_{1\leq n\leq N}|\gamma_n|$. If
\begin{align}
\label{eq:ex}%
\Gamma \leq \frac{1}{2N} \sin^2\left(\frac{\pi}{2(2N+1)}\right),
\end{align}
then all eigenvalues of matrix  $J$ are real.
\end{theorem}

\begin{proof}
We recall that the spectrum of the matrix $H$ is well known. It
consists of   $2N$ distinct  eigenvalues  which can be listed in
the descending order as follows:
$$
\tE_n = 2\cos\left(\frac{n\pi}{2N+1}\right), \quad 1\leq n\leq 2N.
$$
This allows us to estimate the smallest distance between the adjacent
eigenvalues:
\begin{eqnarray}
\min_{1\leq n\leq 2N-1}|\tE_{n+1}-\tE_{n}| & = &
4\min_{1\leq n\leq 2N-1}\left| \sin\left(\frac{2n+1}{2N+1}\frac{\pi}{2}\right)\right|
\sin\left(\frac{\pi}{2(2N+1)}\right) \nonumber \\
& = &  4 \sin\left(\frac{3 \pi}{2(2N+1)}\right) \sin\left(\frac{\pi}{2(2N+1)}\right) \nonumber \\
& \geq & 4\sin^2\left(\frac{\pi}{2(2N+1)}\right) =: 2 r_N.
\label{eq:mindist}%
\end{eqnarray}
Next we introduce the diagonal matrix  of the eigenvalues:
$\hE=$diag$(\tE_1,...,\tE_{2N})$ and the matrix $S$ whose columns
are the eigenvectors of $H$:
$$
S_{nm}=\sin\left(\frac{nm\pi}{2N+1}\right), \quad 1\leq n, m\leq 2N
$$
(i.e. $S$ is the matrix of discrete sine transform). The inverse
matrix is given as
\begin{equation}
S^{-1} = \frac{2}{2N+1}S.
\end{equation}
Thus there exists the similarity transformation: $S^{-1}HS=\hE$.

Let us now apply transformation $S$ to the matrix $J$:
\begin{equation}
\hJ := S^{-1} J S = \hE + i\hG,
\end{equation}
where $\hG := S^{-1}G S$.   Obviously, all the entries of
matrix $\hG$ are real. It is also easy to estimate
\begin{equation}
\label{eq:Gnm} %
|\hG_{nm}| \leq \frac{2}{2N+1} \max_{1 \leq n \leq N} |\gamma_n|
\sum_{k=1}^{2N} \left| \sin\left(\frac{\pi k n}{2N+1}\right) \sin\left(\frac{\pi m k}{2N+1}\right) \right|
< 2\Gamma.
\end{equation}

Let us now estimate location of the eigenvalues of $J$ by
applying Gershgorin's circle theorem  \cite{Lancaster} to the similar matrix $\hJ$.
Radii $R_n$ and centers $C_n$ of Gershgorin's disks for the matrix $\hJ$ are given as
\begin{eqnarray*}
R_n=  \sum_{m=1,\, m \ne n}^{2N} |(\hG)_{nm}|,  \quad
\label{eq:center}%
 C_n = \tE_n + i (\hG)_{nn}, \quad 1\leq n \leq 2N.
\end{eqnarray*}
Real parts of the disk centers  $C_n$ equal to $\tE_n$. Therefore,
in the complex $z$-plane the  $n$th disk belongs to a  strip  $\tE_n - R_n \leq \textrm{Re\ }z \leq\tE_n+ R_n$.
From Eqs.~(\ref{eq:ex}) and (\ref{eq:Gnm})   we observe that  $R_n<4N\Gamma<r_N$, where $r_N$
is the lower boundary of the  half-distance between the eigenvalues defined in (\ref{eq:mindist}).
Therefore, all the $2N$ strips are disjoint. According to Gershgorin's circle theorem \cite{Lancaster},
in this situation  each disk contains exactly one eigenvalue. Since the real parts
of all the eigenvalues  are different, the spectrum does not contain
complex-conjugated eigenvalues and hence   is purely real because complex eigenvalues
in the spectrum of the $\PT$-symmetric matrix  $J$  (if any) always come in complex-conjugate pairs \cite{BenderNessSuff}.
\end{proof}

\begin{theorem}[on broken $\PT$ symmetry]
 \label{breaking}
 Consider the eigenvalue problem (\ref{lin}).
 \begin{enumerate}
    \item If $\sum_{n=1}^N \gamma_n^2 > 2N-1$, then there exist at least two eigenvalues with nonzero imaginary part.
    \smallskip
    \item If $|\gamma_{N}|>1$ and $|\gamma_n|>2$ for each $n\in \{ 1, 2, \ldots, N-1\}$, then
    there exists no pure real  eigenvalues $\tE$ of the matrix $J$.
\end{enumerate}
\end{theorem}

\begin{proof}
The first claim follows from Proposition~\ref{prop:char}. Indeed,  we
can introduce $\xi := \lambda^2$ and rewrite the
characteristic equation as
$$
P(\lambda) = \xi^N - p_2 \xi^{N-1} - \ldots -p_{2N}=0.
$$
Then $p_2$ is sum of all the roots of the latter
equation. The condition of  the theorem implies that  $p_2<0$. Then
there must exist at least one root $\xi_0$ with negative real part. This
would correspond to two eigenvalues $\lambda = \pm \sqrt{\xi_0}$ with
nonzero imaginary part.

In order to prove the second claim,  we compute the radii
$R_n=\sum_{m\neq n}|J_{nm}|$ of Gershgorin's discs
corresponding to the diagonal elements $J_{nn}= i\gamma_n$, $n=-N+1, -N+2, \ldots, N$.
For the first ($n=-N+1$) and the last ($n=N$) radii,
we obtain $R_{-N+1}=R_{N} = 1$, while for the other radii
($n=-N+2, \ldots, N-1$) we have $R_{n}=2$. Then it follows from
Gershgorin's circle theorem  \cite{Lancaster} that if
the condition (ii) of the theorem is satisfied, then there is no
intersection of Gershgorin's discs centered at  $  i\gamma_n$
with the real axis and hence  no eigenvalues lies on the real
axis.
\end{proof}

Let us now consider a parameter space $\{\gamma_1,...,\gamma_{ M}\}$,
where $ M$ ($1\leq M\leq N$) is a number of independent parameters $\gamma_n$.
As a consequence of Theorem~\ref{exact}, there exists a non-empty domain $D_0$ where
the spectrum is entirely real. Obviously, $D_0$ contains
the origin of the parameter space. The boundary of the domain $D_0$
consists of the points at which $\PT$ symmetry breaking occurs.
On the other hand, it follows from Theorem \ref{breaking}
that there exists a non-empty unbounded domain $D_{\infty}$, where the spectrum contains   complex eigenvalues.

We also notice the dependence of the bound (\ref{eq:ex}) in Theorem \ref{exact} on $N$,
i.e. the importance of the length of the chain for the  sufficient condition of unbroken $\PT$
symmetry.

For the three examples of defect, alternating, and clustered chains
listed in the Introduction, there exists a single gain and loss
parameter $\gamma$ (i.e. now $M=1$), which can be assumed to be non-negative without loss of generality.
In this case, the boundary of $D_0$ degenerates in a single point  $\gamma_{\PT}$,
which is   termed as a $\PT$ symmetry breaking threshold.
The results of the previous studies can be summarized as follows.

\begin{enumerate}
  \item The defect chain  was considered in \cite{Lin09,Joglekar10},
   where it was reported   that  $\gamma_{\PT} = 1$ if either $d=1$ or $d=N$;  and
  $\gamma_{\PT} \propto N^{-1}$ if $d=2, 3, \ldots, N-1$ (recall that
   $d$ is the defect position in the chain).

  \item The alternating chain was considered in \cite{Pelin3,Bar}, where
  it was found that
  $$
  \gamma_{\PT} = 2 \cos \left( \frac{\pi N }{1 + 2N} \right)  = 2 \sin \left( \frac{\pi  }{2(1 + 2N)} \right) \propto N^{-1}.
  $$

  \item The clustered chain was also considered in \cite{Bar} using both analytical and graphical arguments.
  It was found that $\gamma_{\PT} \propto  N^{-2}$. Note that for sufficiently large $N$ the clustered $\PT$-chain
  has a narrower stability interval compared to the other two examples.
\end{enumerate}
Note that the condition  (\ref{eq:ex}) in Theorem \ref{exact} is not sharp for any of the examples above
because the upper bound in (\ref{eq:ex}) behaves like $N^{-3}$ as $N \to \infty$.

\section{The dimer model}
\label{sec:dim}

Here we consider the simplest case, $N=1$, of the $\p \T$-symmetric network (\ref{dnls-general}).
This   case is usually referred as to the dimer model. Setting
$\gamma_1=-\gamma_0=\gamma$ and $\chi_0=\chi_1=\chi$, we rewrite the dimer model explicitly:
\begin{eqnarray}
\left\{ \begin{array}{l} i \dot{q}_0 = q_1 - i\gamma q_0 + \left[(1-\chi) |q_0|^2 + \chi |q_1|^2\right] q_0, \\[2mm]
i \dot{q}_1 = q_0 + i \gamma q_1 + \left[\chi |q_{  {0}}|^2 + (1 - \chi) |q_{ {1}}|^2\right] q_1,
\end{array} \right.
\label{dimer-chi}
\end{eqnarray}
where the overdot denotes the derivative with respect to $t$.
Without loss of generality, we assume that $\gamma \geq 0$.
Our aim is to understand the long-term  dynamics of the nonlinear dimer model (\ref{dimer-chi}).
In particular, we are interested in checking whether the dimer model (\ref{dimer-chi}) is integrable and if there
exist solutions that grow to infinity.

We note that the case $\chi = 0$ is well studied (see e.g.
\cite{Ramezani,Sukhorukov,Pelin4,Susanto,Barashenkov} and references therein).
It is known that a convenient way
of treating the system is to pass to the Stokes variables (used for the dNLS equation for the first time in~\cite{CCE88}):
\begin{eqnarray}
\label{spin}
S_0=|q_0|^2+|q_1|^2,
\quad
S_1=q_0 \bar{q}_1+\bar{q}_0q_1,
\quad
S_2=i(q_0 \bar{q}_1- \bar{q}_0q_1),
\quad
S_3=|q_1|^2-|q_0|^2,
\end{eqnarray}
which satisfy the relation $S_0^2=S_1^2+S_2^2+S_3^2$. By means of straightforward algebra,
from system (\ref{dimer-chi}) we obtain the following set of equations
\begin{eqnarray}
\label{S0}
&& \dot{S}_0= 2\gamma  S_3,
\\
\label{S1}
&& \dot{S}_1= (1-2\chi)  S_2 S_3,
\\
\label{S2}
&& \dot{S}_2= 2  S_3 {-}(1-2\chi) S_1S_3,
\\
\label{S3}
&& \dot{S}_3= 2\gamma S_0  {-} 2S_2.
\end{eqnarray}
Global existence of solutions of the dimer model (\ref{dimer-chi})
follows from Eq. (\ref{S0}) because Gronwall's inequality implies that
\begin{equation}
\label{Gronwall} S_0(t) \leq S_0(0) e^{2 \gamma |t|}, \quad t \in
\mathbb{R}.
\end{equation}
Therefore the question we address is if $S_0 = |q_0|^2 +|q_1|^2$ remains
bounded as $t \to \infty$. To this end, we shall separate the cases
$\chi \neq \frac{1}{2}$ and $\chi = \frac{1}{2}$.

\subsection{The case $\chi \neq \frac{1}{2}$}

We shall prove that for any arbitrarily small but  nonzero $\gamma$ the dimer equations (\ref{dimer-chi}) do have
solutions which become unbounded as $t \to \infty$.

\begin{theorem}
Let $\chi \neq \frac{1}{2}$ and $\gamma > 0$. Then system
(\ref{dimer-chi}) has an unbounded solution as $t\to\infty$. \label{theorem-growth}
\end{theorem}

\begin{proof}
Equations~(\ref{S1}) and (\ref{S2}) can be reduced to the harmonic oscillator equation
in the new temporal variable
\begin{equation}
\label{new-variable-s}
s(t) := \int_0^t S_3(t')dt'.
\end{equation}
Note that the variable $s(t)$ is well defined for all $t \in \mathbb{R}$, since the solution
of the system  (\ref{dimer-chi}) exists globally for all $t \in \mathbb{R}$.
Using the auxiliary variable $s(t)$, we obtain an equivalent representation of
solutions of Eqs.~(\ref{S1}) and (\ref{S2}):
\begin{equation}
\label{exact-ab-variable}
\left\{ \begin{array}{l}
\displaystyle{S_1(t) = \frac{2}{1 - 2\chi} + C_1 \cos\left[ (1-2 \chi) s(t) \right]
+ C_2 \sin \left[(1-2 \chi)s(t)\right],}
\\
\displaystyle{S_2(t) = -C_1 \sin\left[(1-2 \chi) s(t)\right] + C_2
\cos \left[(1-2 \chi) s(t) \right],}
\end{array} \right.
\end{equation}
where $C_1$ and $C_2$ are   constants, which are uniquely defined
by the initial conditions:
\begin{equation}
\label{exact-ab-variable-initial}
C_1 = S_1(0) - \frac{2}{1-2\chi}, \quad C_2 = S_2(0).
\end{equation}
Using Eqs. (\ref{exact-ab-variable}) and (\ref{exact-ab-variable-initial}), one can easily find
constants $D$ and $E$ independent on $q_0(0)$ and $q_1(0)$ such that
\begin{equation}
\label{bound-ab} |q_0(t) q_1(t)| \leq \frac{1}{|1 - 2\chi|} +
|C_1| + |C_2| <  D |q_0(0)q_1(0)|  + E,
\end{equation}
hence $|q_0(t) q_1(t)|$ is a bounded function of $t$ for all times.

To complete the proof, we multiply the second equation of system
(\ref{dimer-chi}) by $\bar{q}_1$, and add its complex conjugate equation
to obtain
\begin{equation}
\label{second-eq-amplitude}
 \frac{d |q_1|^2}{d t} = 2\gamma |q_1|^2 - i( \bar{q}_1 q_0 - q_1
 \bar{q}_0) > 2\gamma |q_1|^2 - 2( D |q_0(0)q_1(0)|  + E),
\end{equation}
where the latter inequality follows from Eq.~(\ref{bound-ab}).
Let us now introduce the lower solution $|{q}_L(t)|^2$ that satisfies
the differential equation
\begin{equation}
\label{eq:lower}
 \frac{d |{q_L}|^2}{d t} = 2 \gamma |{q_L}|^2 - 2 \left(D|q_0(0)q_1(0)| +  E \right).
\end{equation}
Then using (\ref{second-eq-amplitude}) we have
\begin{equation}
\label{eq:chain}
0= \frac{d |{q_L}|^2}{d t} - 2 \gamma |{q_L}|^2 + 2 \left(D|q_0(0)q_1(0)| +  E \right) <
\frac{d |{q_1}|^2}{d t} - 2 \gamma |{q_1}|^2 + 2 \left(D|q_0(0)q_1(0)| +  E \right).
\end{equation}
Let us choose
the initial data $q_1(0)$  to be sufficiently large
and  $q_0(0)$ to be sufficiently small such that
\begin{equation}
\label{eq:chooseql}
\frac{D|q_0(0)q_1(0)| +  E}{\gamma} < |q_L(0)|^2 \leq |q_1(0)|^2.
\end{equation}
The first inequality in Eq.~(\ref{eq:chooseql}) (together with  Eq.~(\ref{eq:lower}))
implies that the lower solution grows like $e^{2\gamma t}$ for positive $t$.
Using the second inequality in Eq.~(\ref{eq:chooseql}) as well as
inequality in Eq.~(\ref{eq:chain}), we apply the comparison theorem
for differential equations~\cite{Walter} and prove that
$|q_L(t)|^2 \leq |q_1(t)|^2$ for all $t>0$. Therefore, $|q_1(t)|^2$
grows at least exponentially as $t \to \infty$.
\end{proof}

\subsection{The case $\chi = \frac{1}{2}$}
\label{chi12}

It follows from (\ref{S1}) that now $\dot{S}_1=0$, i.e. $S_1$ is an
integral of motion. Moreover, from Eqs. (\ref{S0}) and (\ref{S2}) we obtain that
\begin{equation}
\label{conserved} Q := S_0 -  {\gamma}S_2
\end{equation}
is also an integral of motion: $\dot{Q} = 0$. This allows us to prove the following
result.

\begin{theorem}
Let $\chi = \frac{1}{2}$. Then for $\gamma \in [0,1)$ all solutions of (\ref{dnls-general}) are bounded.
If $\gamma \geq  1$ then there exist unbounded solutions. \label{theorem-global}
\end{theorem}

\begin{proof}
From Eqs.~(\ref{S0}), (\ref{S3}), and (\ref{conserved}) we obtain that
\begin{equation}
\label{harmonic}
\ddot{S_0} + 4(1 - \gamma^2) S_0 = 4 Q.
\end{equation}
If $\gamma \in [0,1)$ then all solutions
of the harmonic oscillator equation (\ref{harmonic}) with constant $Q$ are bounded
for all times.  If  $\gamma \geq  1$, then there exist growing
solution of the linear equation (\ref{harmonic}) with the polynomial growth if $\gamma = 1$
and the exponential growth if $\gamma > 1$.
\end{proof}

The existence  of  two integrals of motion ($S_1$ and $Q$) is an indication of full
integrability of the dimer model (\ref{dimer-chi}) with $\chi=\frac{1}{2}$.     Note that the dimer model
(\ref{dimer-chi}) with $\chi = 0$ also has two conserved quantities \cite{Ramezani}
but these quantities are different from $S_1$ and $Q$. In particular,
the two conserved quantities for $\chi = 0$ do not prevent the solutions of the dimer model
to grow unboundedly \cite{Pelin4,Susanto,Barashenkov} as in Theorem~\ref{theorem-growth}.

\section{Existence and classification of stationary nonlinear  modes}
\label{sec:class}

In this section we consider existence and classification of stationary nonlinear modes
of the system of algebraic equations (\ref{stationary-general})  with arbitrary $N$ in the limit of large $E$.
Let us first recall that small-amplitude nonlinear modes bifurcating from linear modes of the linear eigenvalue equation (\ref{lin})
were considered in \cite{ZK} for simple eigenvalues and in \cite{ZK2} for semi-simple and double eigenvalues.
The arguments of \cite{ZK} were rigorously justified with the method of Lyapunov--Schmidt reductions
in \cite{Pelin3}. We shall now consider the opposite limit of large amplitudes
(also referred to as the anti-continuum limit), where all possible stationary nonlinear modes can be
fully classified.

First in Theorem~\ref{theorem-general}, we specify a particular result on the existence of stationary modes when the
amplitudes of all $2N$ sites  are large in the limit of $E \to \infty$. Then in Theorem~\ref{theorem-nonlocal-bifurcation},
we give a general result on existence of stationary modes when some amplitudes become zero in this limit.
At last, in Proposition \ref{theorem-no-solutions}, we
rule out existence of any other stationary modes for sufficiently large $E$. This allows us to find
in Proposition \ref{theorem-exactly} the exact number of unique (up to a gauge transformation) nonlinear
modes existing in the limit of large $E$.

To enable the consideration of the limit $E \to \infty$,  we rescale the variables
$E = 1/\delta$ and ${\bf w} = {\bf W}/\delta^{1/2}$ with small positive $\delta$ and rewrite
the system of algebraic equations (\ref{stationary-general}) in the form
\begin{equation}
\label{anti-continuum}
(1 - |W_n|^2) W_n = \delta( W_{n+1} + W_{n-1} + i \gamma_n W_n), \quad 1 \leq n \leq N,
\end{equation}
subject to the boundary conditions $W_0 = \bar{W}_1$ and $W_{N+1} = 0$.

\begin{theorem}
\label{theorem-general}
For any given $N \in \mathbb{N}$, let coefficients
$\gamma_1, \gamma_2, \ldots \gamma_{N}$ satisfy the constraints
\begin{equation}
\label{condition}
\left|\sum_{n=K}^N \gamma_n \right|<1 \quad \mbox{for each } K = 1,2,\ldots, N.
\end{equation}
Then the nonlinear stationary equations (\ref{stationary-general})
admit  $2^N$ $\PT$-invariant, i.e. satisfying ${\bf w} = \p \bar{\bf w}$,
solutions (unique up to a gauge transformation) in the limit of large positive $E$ such that,
for sufficiently large $E$, the map $E \to {\bf w}$ is $C^{\infty}$ at each solution
and there is a positive $E$-independent constant $C$ such that
\begin{equation}
\label{bound-nonlocal}
\left| |w_n|^2 -  E \right| \leq C \quad \mbox{for each } n = 1,2, \ldots, N.
\end{equation}
\end{theorem}

\begin{proof}
Separating the amplitude and phase variables in the rescaled system (\ref{anti-continuum}),
\begin{equation}
\label{parameter-1}
W_n = \left( \prod_{j=1}^n A_j^{1/2} \right) e^{i \sum_{j=1}^n \varphi_j}, \quad \mbox{\rm for each} \; n = 1,2,...,N,
\end{equation}
we obtain $N$ equations for phases
\begin{equation}
\left\{ \begin{array}{l} A_2^{1/2} \sin(\varphi_2) - \sin(2 \varphi_1) + \gamma_1 = 0, \\
A_{n+1}^{1/2} \sin(\varphi_{n+1}) - A_n^{-1/2} \sin(\varphi_n) + \gamma_n = 0, \quad 2 \leq n \leq N-1, \\
- A_{N}^{-1/2} \sin(\varphi_{N}) + \gamma_N = 0,  \end{array} \right.
\label{parameter-2}
\end{equation}
and $N$ equations for amplitudes
\begin{equation}
\label{parameter-3}
\left\{ \begin{array}{l} 1 - A_1 = \delta \left( A_2^{1/2} \cos(\varphi_2) + \cos(2 \varphi_1)\right), \\
1 - \prod_{j=1}^n A_j = \delta \left( A_{n+1}^{1/2} \cos(\varphi_{n+1}) + A_{n}^{-1/2} \cos(\varphi_{n}) \right), \quad 2 \leq n \leq N-1, \\
1 - \prod_{j=1}^N A_j = \delta A_{N}^{-1/2} \cos(\varphi_{N}).
\end{array} \right.
\end{equation}
By the Implicit Function Theorem, for small $\delta$ and any $(\varphi_1,\varphi_2,...,\varphi_N)$,
there exists a unique solution of system (\ref{parameter-3})
for the amplitude variables with $A_n = 1 + \mathcal{O}(\delta)$ for each $n = 1,2,...,N$.
Substituting this solution to
system (\ref{parameter-2}), we obtain at the leading order
\begin{equation}
\left\{ \begin{array}{l} \sin(\varphi_2) - \sin(2 \varphi_1) + \gamma_1 = 0, \\
\sin(\varphi_{n+1}) -\sin(\varphi_{n}) + \gamma_{n} = 0, \quad 2 \leq n \leq N-1, \\
- \sin(\varphi_{N}) + \gamma_N = 0.  \end{array} \right.
\label{parameter-4}
\end{equation}
Solving the system backwards, we obtain
\begin{equation}
\left\{ \begin{array}{l}
\sin(\varphi_N) = \gamma_N, \\
\sin(\varphi_{N-n}) = \sum_{j=N-n}^{N} \gamma_j, \quad 1 \leq n \leq N-2, \\
\sin(2 \varphi_1) = \sum_{j=1}^{N}\gamma_j.
\end{array} \right.
\label{parameter-4a}
\end{equation}
There exists $2^N$ solutions of equations (\ref{parameter-4a})
for the phase variables $\varphi_N$, $\varphi_{N-1}$, ..., $\varphi_2$, and $2 \varphi_1$
in the interval $\left(-\frac{\pi}{2}, \frac{3\pi}{2}\right)$, provided the constraints (\ref{condition}) are satisfied.
Note that the found $2^N$ solutions have $\varphi_1 \in \left(-\frac{\pi}{4}, \frac{3\pi}{4}\right)$. There exist other $2^N$
solutions with $\varphi_1 \in \left(\frac{3\pi}{4}, \frac{7\pi}{4}\right)$.
However, the latter $2^N$    solutions
can be obtained by the gauge transformation: if ${\bf W}$ is a solution,
then $-{\bf W}$ is also a solution of the nonlinear system (\ref{anti-continuum}).

Each of the $2^N$ solutions of the system (\ref{parameter-4}) is non-degenerate in the sense
that the Jacobian matrix is upper-triangular and non-singular under the same constraints (\ref{condition}).
Therefore, each solution is uniquely continued
with respect to small parameter $\delta$. Hence, we have the existence of $2^N$ finite-amplitude
solutions of the stationary equations (\ref{anti-continuum}). Substituting
the scaling transformation, we obtain bound (\ref{bound-nonlocal})
for solutions of the stationary equations (\ref{stationary-general}).
\end{proof}

\begin{theorem}
\label{theorem-nonlocal-bifurcation}
For any given $N \in \{2,3,\ldots\}$, fix $M \in \{1,2,\ldots,N-1\}$.
Let coefficients $\gamma_1, \gamma_2, \ldots \gamma_{M}$ satisfy the constraints
\begin{equation}
\label{condition-modified}
\left|\sum_{n=K}^M \gamma_n \right|<1 \quad \mbox{for each } K = 1,2,\ldots, M.
\end{equation}
Then the nonlinear stationary equations (\ref{stationary-general}) admit
$2^M$ $\PT$-invariant solutions ${\bf w} = \p \bar{\bf w}$
(unique up to a gauge transformation) in the limit of large positive $E$ such that,
for sufficiently large $E$, the map $E \to {\bf w}$ is $C^{\infty}$ at each solution
and there is a positive $E$-independent constant $C$ such that
\begin{equation}
\begin{array}{ll}
\label{bound-nonlocal-modified}
\left| |w_n|^2 -  E \right| \leq  C  &\mbox{for each } n = 1,2, \ldots, M,\\[2mm]
 |w_n|^2  \leq   C E^{-1} &\mbox{for each } n = M+1, M+2, \ldots, N.
\end{array}
\end{equation}
\end{theorem}

\begin{proof}
Compared to the proof of Theorem \ref{theorem-general},
we modify the decomposition (\ref{parameter-1}) in the form
\begin{equation}
\label{parameter-5}
W_n = \left( \prod_{j=1}^n A_j^{1/2} \right) e^{i \sum_{j=1}^n \varphi_j}, \quad \mbox{\rm for each} \; n = 1,2,...,M,
\end{equation}
and leave the variables $\{ W_{M+1},\ldots,W_N\}$ unaffected. Then, equations
(\ref{anti-continuum}) are rewritten in phase and amplitude variables for $1 \leq n \leq M$
and left unchanged for $M+1 \leq n \leq N$. Variables $\{ W_n \}_{M+1}^N$ are considered
in the neighborhood of the zero equilibrium for small values of $\delta$. By the Implicit Function Theorem,
for small $\delta$ and any given $W_M$, there exists a unique solution of system (\ref{anti-continuum})
for $M+1 \leq n \leq N$ such that
\begin{equation}
|W_n| \leq C \delta |W_M|, \quad M+1 \leq n \leq N,
\end{equation}
for some positive $(\delta,W_M)$-independent constant $C$.
The nonlinear systems (\ref{parameter-2}) and (\ref{parameter-3}) are now written in the form
\begin{equation*}
\left\{ \begin{array}{l} A_2^{1/2} \sin(\varphi_2) - \sin(2 \varphi_1) + \gamma_1 = 0, \\
A_{n+1}^{1/2} \sin(\varphi_{n+1}) - A_n^{-1/2} \sin(\varphi_n) + \gamma_n = 0, \quad 2 \leq n \leq M-1, \\
{\rm Im}\left(W_{M+1} (\prod_{j=1}^M A_j)^{-1/2} e^{-i \sum_{j=1}^M \varphi_j}\right)
- A_{M}^{-1/2} \sin(\varphi_{M}) + \gamma_M = 0,  \end{array} \right.
\end{equation*}
and
\begin{equation*}
\left\{ \begin{array}{l} 1 - A_1 = \delta \left( A_2^{1/2} \cos(\varphi_2) + \cos(2 \varphi_1)\right), \\
1 - \prod_{j=1}^n A_j = \delta \left( A_{n+1}^{1/2} \cos(\varphi_{n+1}) + A_{n}^{-1/2} \cos(\varphi_{n}) \right), \quad 2 \leq n \leq M-1, \\
1 - \prod_{j=1}^M A_j = \delta \left({\rm Re}\left(W_{M+1} (\prod_{j=1}^M A_j)^{-1/2} e^{-i\sum_{j=1}^M \varphi_j}\right)
+ A_{M}^{-1/2} \cos(\varphi_{M}) \right).
\end{array} \right.
\end{equation*}
The system of equations for phase and amplitude variables can be studied similarly to the system
(\ref{parameter-2}) and (\ref{parameter-3}), where the additional
terms with the variable $W_{M+1}$ are found to be $\mathcal{O}(\delta)$ small as $\delta \to 0$.
Substituting the scaling transformation, we obtain bound (\ref{bound-nonlocal-modified})
for solutions of the stationary equations (\ref{stationary-general}).
\end{proof}

\begin{proposition}
For any given $N \in \mathbb{N}$, let coefficients $\gamma_1, \gamma_2, \ldots \gamma_N$
satisfy the constraints (\ref{condition}) and additional constraints
\begin{equation}
\label{condition-modified-new}
\sum_{n=K}^M \gamma_n \neq 0, \quad \mbox{for each } K = 2,3,\ldots,M\; \mbox{and each} \; M = 2,3,\ldots,N.
\end{equation}
Then, besides stationary solutions of Theorems \ref{theorem-general} and \ref{theorem-nonlocal-bifurcation},
no other stationary solutions of system (\ref{stationary-general}) exist for sufficiently large $E$.
\label{theorem-no-solutions}
\end{proposition}

\begin{proof}
All stationary solutions of  Theorems \ref{theorem-general} and \ref{theorem-nonlocal-bifurcation}
are characterized by a nonzero limit of $W_n$ for any $1 \leq n \leq M$
and zero limit of $W_n$ for any $M+1 \leq n \leq N$ as $\delta \to 0$,
where $M = 1,2,...,N$.
For any other possible solution, there must exist $K$ in $1 \leq K \leq M-1$ for $M \geq 2$  such that
$W_K \to 0$ as $\delta \to 0$. We will now show that this solution cannot be continued in $\delta$ under
the conditions (\ref{condition-modified-new}). Indeed, the persistence analysis of
Theorems \ref{theorem-general} and \ref{theorem-nonlocal-bifurcation} would result
in the following set of bifurcation equations at $\delta = 0$:
\begin{equation}
\left\{ \begin{array}{l} \sin(\varphi_{K+2})  + \gamma_{K+1} = 0, \\
\sin(\varphi_{n+1}) -\sin(\varphi_{n}) + \gamma_{n} = 0, \quad K+2 \leq n \leq M-1, \\
- \sin(\varphi_{M}) + \gamma_M = 0,  \end{array} \right.
\label{parameter-4-new}
\end{equation}
from which we realize that we have one equation more than the number of phase variables.
No solution exists under the condition (\ref{condition-modified-new}).
\end{proof}

\begin{proposition} Under the conditions of Proposition \ref{theorem-no-solutions},
system (\ref{stationary-general}) has  exactly  $2^{N+1}-2$ $\PT$-invariant
stationary solutions (unique up to a gauge transformation) for sufficiently large $E$.
\label{theorem-exactly}
\end{proposition}

\begin{proof}
We note that all the $2^N$ solutions of  Theorems \ref{theorem-general} extend the $2^M$ solutions
of  Theorem \ref{theorem-nonlocal-bifurcation} for $M = N$. Therefore,
Theorems  \ref{theorem-general} and \ref{theorem-nonlocal-bifurcation}  yield
$\sum_{M=1}^N 2^M = 2^{N+1}-2$  distinct  solutions of the algebraic equations
(\ref{stationary-general}) for sufficiently large $E$. No other solution exists
due to Proposition \ref{theorem-no-solutions}.
\end{proof}

In the context of Proposition \ref{theorem-no-solutions}, we should remark that even if
the conditions (\ref{condition-modified-new}) are not satisfied for some $K$ and $M$,
that is, if $\sum_{n=K}^M \gamma_n = 0$, we do not generally anticipate bifurcations
of additional stationary solutions because the Jacobian of the system of nonlinear equations
(\ref{parameter-4-new}) is singular and the implicit function theorem cannot be used to
continue phase variables in $\delta$. Nevertheless, there may exist special configurations
of the dNLS-type network (\ref{dnls-general}), where conditions (\ref{condition-modified-new}) are
not satisfied and additional degenerate branches of stationary modes exist in the system
of nonlinear algebraic equations (\ref{anti-continuum}) for small $\delta$.

For the three examples of defect, alternating, and clustered $\PT$-symmetric chains
listed in the introduction, the constraints (\ref{condition})
acquire a much simpler form. In particular, for the defect and alternating chains,
the constraints yield $|\gamma|<1$, whereas for the clustered chain, they yield
$|\gamma| < N^{-1}$.

For the defect and alternating chains, all solutions of Theorems  \ref{theorem-general} and \ref{theorem-nonlocal-bifurcation}
exist uniformly for $|\gamma| < 1$ and disappear if $|\gamma| > 1$.
For the clustered chain on the other hand, we observe that the interval $(-N^{-1},N^{-1})$
converge to zero as $N \to \infty$ but it converges much slower than
the stability interval $(-\gamma_{\PT},\gamma_{\PT})$ of the zero equilibrium  with
$\gamma_{\PT}   \propto N^{-2}$. Therefore, even if some families of the stationary solutions
do not exist in the small-amplitude limit if $N^{-2} < |\gamma| < N^{-1}$,
there exist at least $2^N$ families of Theorem \ref{theorem-general}
in the large-amplitude limit.
For $|\gamma| > N^{-1}$, all these $2^N$ families disappear in the large-amplitude limit.
Additional $2^M$ branches of Theorem \ref{theorem-nonlocal-bifurcation}
exist if $N^{-1} < |\gamma| < M^{-1}$ for $M \in \{1,2,\ldots,N-1\}$.

\section{Stability of stationary nonlinear modes}
\label{sec:stab}

To consider stability of the stationary nonlinear modes, we employ
the substitution
$$
{\bf q}(t) = e^{-iEt} ({\bf w} + {\bf U}(t))
$$
and linearize the nonlinear $\PT$-dNLS equations (\ref{dnls-general})
at the $\p\T$-invariant modes ${\bf w}$ with $w_{1-n} = \bar{w}_n$.
As a result, we obtain the linearized time-evolution problem
\begin{equation}
\begin{array}{l}
i \frac{d U_n}{dt} + E U_n = U_{n+1} + U_{n-1} + i \gamma_n U_n \\
\phantom{texttext} +
(2-\chi_n) |w_n|^2 U_n + (1-\chi_n) w_n^2 \bar{U}_n +
\chi_n w_n^2 U_{1-n} + \chi_n |w_n|^2 \bar{U}_{1-n}.
\end{array} 
\end{equation}
Then, singling out the spectral parameter in the
system of two equations for
$$
{\bf U}(t) = \mbox{\boldmath $\phi$} e^{-\lambda t}  \quad
\mbox{\rm and} \quad \bar{{\bf U}}(t) = \mbox{\boldmath $\psi$}
e^{-\lambda t},
$$
we arrive at the spectral problem
\begin{equation}
\label{spectral-problem}
\left\{ \begin{array}{l}
(E-i \lambda) \phi_n = \phi_{n+1} + \phi_{n-1} + i \gamma_n \phi_n \\
\phantom{texttext} +
(2-\chi_n) |w_n|^2 \phi_n + (1-\chi_n) w_n^2 \psi_n + \chi_n w_n^2 \phi_{1-n}
+ \chi_n |w_n|^2 \psi_{1-n}, \\
(E+i \lambda) \psi_n = \psi_{n+1} + \psi_{n-1} - i \gamma_n \psi_n +
(1-\chi_n) \bar{w}_n^2 \phi_n \\
\phantom{texttext} + (2-\chi_n) |w_n|^2 \psi_n
+ \chi_n \bar{w}_n^2 \psi_{1-n} + \chi_n |w_n|^2 \phi_{1-n}.
\end{array} \right.
\end{equation}
The system is truncated at the sites
$n \in \{-N+1,\ldots,N\}$ subject to the Dirichlet boundary
conditions at $n = -N$ and $n = N+1$. Therefore, we have exactly $4N$ eigenvalues
in the spectral problem (\ref{spectral-problem}). We emphasize that
unless $\lambda$ is real, $\phi_n$ and $\psi_n$ are not complex-conjugate to each other.

We say that the stationary nonlinear mode is spectrally stable if
the spectral problem (\ref{spectral-problem}) has no eigenvalues $\lambda$
with ${\rm Re}(\lambda) > 0$. By the $\p\T$ symmetry of the stationary mode ${\bf w}$
and the related symmetry of the eigenvectors $(\mbox{\boldmath $\phi$}, \mbox{\boldmath $\psi$})$
in the linear system (\ref{spectral-problem}),
all eigenvalues $\lambda$ are symmetric about the origin
and with respect to the complex conjugation. As a result,
all eigenvalues have ${\rm Re}(\lambda) = 0$ in the case of spectral stability
of the  stationary mode ${\bf w}$.

\subsection{Analytical results}
\label{sec:stab-analyt}

Here we obtain classification of stability of the stationary nonlinear modes constructed in section~\ref{sec:class} in the limit of large $E$.
Theorem \ref{theorem-stability-1} addresses stability of the stationary modes
in Theorem \ref{theorem-general}, whereas Theorem \ref{theorem-stability-2}
elaborates stability of a particular dimer-type ($M = 1$) stationary modes
in Theorem \ref{theorem-nonlocal-bifurcation}. A more general result
for the stationary solutions of Theorem \ref{theorem-nonlocal-bifurcation}
with $M$ between $2$ and $N-1$ can be obtained based on the count of eigenvalues in
Theorems \ref{theorem-stability-1} and \ref{theorem-stability-2} but
we do not formulate the general result. Thanks to the symmetry of the
$\p\T$-symmetric lattice (\ref{dnls-general}) in the case when all $\chi_n = \frac{1}{2}$,
there are more spectrally stable stationary modes in this case, compared
to the case when all $\chi_n \neq \frac{1}{2}$. Again, the case when some
$\chi_n = \frac{1}{2}$ and some other $\chi_n \neq \frac{1}{2}$ can be considered
separately but we avoid a lengthy formulation of the general result.

\begin{theorem}
\label{theorem-stability-1}
Consider $2^N$ stationary solutions of Theorem \ref{theorem-general} under the constraints
(\ref{condition}). If $\chi_n < \frac{1}{2}$ for all $n$, then
there exists exactly one spectrally stable stationary mode among the $2^N$ solutions
for sufficiently large $E$. If $\chi_n = \frac{1}{2}$ for all $n$, then
there exist exactly two spectrally stable stationary modes among the $2^N$ solutions
for sufficiently large $E$.
\end{theorem}

\begin{proof}
Using the rescaling
$E = 1/\delta$, ${\bf w} = {\bf W}/\delta^{1/2}$, and $\lambda = \Lambda/\delta$
with small positive $\delta$, we rewrite the
spectral problem (\ref{spectral-problem}) in the form
\begin{equation}
\label{spectral-problem-E}
\left\{ \begin{array}{l}
(1 - (2-\chi_n)|W_n|^2) \phi_n - (1 - \chi_n) W_n^2 \psi_n -
\chi_n W_n^2 \phi_{1-n} - \chi_n |W_n|^2 \psi_{1-n} \\
\phantom{texttext} = i \Lambda \phi_n + \delta \left( \phi_{n+1} + \phi_{n-1} + i \gamma_n \phi_n \right), \\
-(1-\chi_n) \bar{W}_n^2 \phi_n + (1 - (2-\chi_n)|W_n|^2) \psi_n  - \chi_n |W_n|^2 \phi_{1-n} - \chi_n \bar{W}_n^2 \psi_{1-n}  \\
\phantom{texttext} = -i \Lambda \psi_n + \delta \left( \psi_{n+1} + \psi_{n-1} - i \gamma_n \psi_n \right).
\end{array} \right.
\end{equation}
Recall from the construction of Theorem \ref{theorem-general} that
$W_n = e^{i \theta_n} \left(1 + \mathcal{O}(\delta) \right)$
for $1 \leq n \leq N$ as $\delta \to 0$, where $\theta_1 = \varphi_1$ and $\theta_{n} - \theta_{n-1} = \varphi_n$ for $2 \leq n \leq N$.
We also recognize from the system (\ref{parameter-4a}) for the phase variables that
$2 \varphi_1$ and each $\varphi_n$ for $2 \leq n \leq N$ is taken either in
the interval $\left(-\frac{\pi}{2},\frac{\pi}{2}\right)$ or in the interval $\left(\frac{\pi}{2},\frac{3\pi}{2}\right)$ by the binary
solution of the equations for the sine-functions provided the condition (\ref{condition})
is satisfied. We say that
the adjacent $n^{\rm th}$ and "$(n-1)$th and $n$th nodes of the dNLS lattice are {\em in-phase}
if $\varphi_n \in \left(-\frac{\pi}{2},\frac{\pi}{2}\right)$
or {\em out-of-phase} if $\varphi_n \in \left(\frac{\pi}{2},\frac{3\pi}{2}\right)$.
Taking into account
that $W_{-n+1} = \bar{W}_n$, we say that
the $0^{\rm th}$ and $1^{\rm st}$ nodes of the dNLS lattice are {\em in-phase}
if $2 \varphi_1 \in \left(-\frac{\pi}{2},\frac{\pi}{2}\right)$
or {\em out-of-phase} if $2 \varphi_1 \in \left(\frac{\pi}{2},\frac{3\pi}{2}\right)$.

\vspace{0.2cm}

{\bf Case $\chi_n < \frac{1}{2}$ for all $n$:} We will show that the only spectrally
stable configuration of the stationary nonlinear modes is the one with all
adjacent nodes being out-of-phase. This conclusion fully agrees with the stability
theorem for discrete solitons in the focusing dNLS equation \cite{pkf}. In fact,
we intend to show that the proof of Theorem \ref{theorem-stability-1} reduces
to the computations of the first-order perturbation theory from the previous
analysis in \cite{pkf}. Justification of the first-order perturbation theory can be
found in Chapter 4.3 of the monograph \cite{Pel-book}.

For $\delta = 0$, there exists only one eigenvalue $\Lambda = 0$ in the
spectral problem (\ref{spectral-problem-E}). If all $\chi_n \neq \frac{1}{2}$,
then the zero eigenvalue has geometric multiplicity $2N$
with the kernel spanned by the eigenvectors
\begin{equation}
\label{eigenvectors}
\left[ \begin{array}{l} \phi_n \\ \psi_n \end{array} \right] = i \left[ \begin{array}{l}
e^{i \theta_n} \\ - e^{-i \theta_n} \end{array} \right], \quad -N+1 \leq n \leq N,
\end{equation}
where $\theta_{-n+1} = -\theta_n$ for $1 \leq n \leq N$, thanks to the $\p\T$ symmetry of
the stationary mode. For each eigenvector in (\ref{eigenvectors}), there exists a generalized
eigenvector
\begin{equation}
\label{eigenvectors-generalized}
\left[ \begin{array}{l} \phi_n \\ \psi_n \\ \phi_{1-n} \\ \psi_{1-n} \end{array} \right] =
\frac{1}{2 (1 - 2 \chi_n)} \left[ \begin{array}{l}
(1-\chi_n) e^{i \theta_n} \\ (1 - \chi_n) e^{-i \theta_n} \\
-\chi_n e^{-i \theta_n} \\ -\chi_n e^{i \theta_n} \end{array} \right], \quad -N+1 \leq n \leq N,
\end{equation}
which satisfies a derivative of the spectral problem (\ref{spectral-problem-E}) in $\Lambda$.

For the first-order perturbation theory, we construct the perturbation expansion
$\Lambda = \mu \delta^{1/2} + \mathcal{O}(\delta^{3/2})$ and a linear superposition
of the eigenvectors (\ref{eigenvectors}) and the generalized eigenvectors (\ref{eigenvectors-generalized}),
\begin{eqnarray}
\nonumber
\left[ \begin{array}{l} \phi_n \\ \psi_n \\ \phi_{1-n} \\ \psi_{1-n} \end{array} \right] & = &
i c_n \left[ \begin{array}{l} e^{i \theta_n} \\ - e^{-i \theta_n} \\ 0 \\ 0 \end{array} \right] +
i c_{1-n} \left[ \begin{array}{l} 0 \\ 0 \\ e^{-i \theta_n} \\ - e^{i \theta_n} \end{array} \right]
+ \frac{\mu \delta^{1/2} c_n}{2 (1 - 2 \chi_n)} \left[ \begin{array}{l}
(1-\chi_n) e^{i \theta_n} \\ (1 - \chi_n) e^{-i \theta_n} \\
-\chi_n e^{- i \theta_n} \\ -\chi_n e^{i \theta_n} \end{array} \right]
\\
& \phantom{t} &
+ \frac{\mu \delta^{1/2} c_{1-n}}{2 (1 - 2 \chi_n)}  \left[ \begin{array}{l}
-\chi_n e^{i \theta_n} \\ -\chi_n e^{- i \theta_n} \\
(1-\chi_n) e^{-i \theta_n} \\ (1-\chi_n) e^{i \theta_n} \end{array} \right]
+\delta \left[ \begin{array}{l} \phi_n^{(1)} \\ \psi_n^{(1)} \\ \phi_{1-n}^{(1)} \\ \psi_{1-n}^{(1)} \end{array} \right]
+ \mathcal{O}(\delta^{3/2}),\label{perturbation}
\end{eqnarray}
where $\mu$ is a new eigenvalue and $\{ c_n \}_{-N+1}^N$ are components of the eigenvector.

To derive equations at the first order that uniquely determine all $4N$ small eigenvalues $\Lambda$,
we represent $W_n = e^{i \theta_n} a_n^{1/2}$, where $a_n$ is expanded from the
system of amplitude equations (\ref{parameter-3}) as $a_n = 1 - \delta a_n^{(1)} + \mathcal{O}(\delta^2)$,
where
\begin{equation*}
\left\{ \begin{array}{l} a_1^{(1)} = \cos(\varphi_2) + \cos(2 \varphi_1), \\
a_n^{(1)} = \cos(\varphi_{n+1}) + \cos(\varphi_n), \quad 2 \leq n \leq N-1, \\
a_N^{(1)} = \cos(\varphi_{N}).
\end{array} \right.
\end{equation*}
Note that due to the definition of $\theta_n$, these equations can be written in the form
\begin{equation}
\label{amplitude-expansion}
a_n^{(1)} = \cos(\theta_{n+1}-\theta_n) + \cos(\theta_n - \theta_{n-1}), \quad 1 \leq n \leq N,
\end{equation}
subject to the boundary conditions $\theta_0 = -\theta_1$ and $\theta_{N+1} = \theta_N + \frac{\pi}{2}$.
This representation is independent of the choice for the coefficients $\gamma_n$ and $\chi_n$.
Moreover, we can extend this representation to $-N+1 \leq n \leq 0$ by $a_{1-n}^{(1)} = a_n^{(1)}$
with the natural definition $\theta_{-N} = \theta_{-N+1} - \frac{\pi}{2}$.

Substituting Eq. (\ref{perturbation}) to the spectral problem
(\ref{spectral-problem-E}), we obtain the system of difference equations:
\begin{eqnarray*}
& \phantom{t} &    (\chi_n - 1) \left( e^{-i \theta_n} \phi_n^{(1)} + e^{i \theta_n} \psi_n^{(1)} \right)
- \chi_n \left( e^{-i \theta_n} \psi_{1-n}^{(1)} + e^{i \theta_n} \phi_{1-n}^{(1)} \right) =   - \gamma_n c_n - i a_n^{(1)} c_n
\\& \phantom{t} & %
\phantom{text}  +\frac{i\mu^2}{2 (1-2\chi_n)}  \left[ (1-\chi_n) c_n - \chi_n c_{1-n} \right]
 +i \left( c_{n+1} e^{i (\theta_{n+1}-\theta_n)} + c_{n-1} e^{i (\theta_{n-1}-\theta_n)} \right),\\[3mm]%
& \phantom{t} & (\chi_n - 1) \left( e^{-i \theta_n} \phi_n^{(1)} + e^{i \theta_n} \psi_n^{(1)} \right)
- \chi_n \left( e^{i \theta_n} \phi_{1-n}^{(1)} + e^{-i \theta_n} \psi_{1-n}^{(1)} \right) = - \gamma_n c_n + i a_n^{(1)} c_n
\\& \phantom{t} & %
\phantom{text}
-\frac{i\mu^2 }{2 (1-2\chi_n)} \left[ (1-\chi_n) c_n - \chi_n c_{1-n} \right]
-i \left( c_{n+1} e^{-i (\theta_{n+1}-\theta_n)} + c_{n-1} e^{-i (\theta_{n-1}-\theta_n)} \right),
\end{eqnarray*}
Note that derivation of these equations is independent from the fact that the phase variable $\theta_n$
depends on $\delta$ as the same phase variables are included in the first two terms of the
decomposition (\ref{perturbation}). Eliminating the first-order correction terms and
using (\ref{amplitude-expansion}) for $a_n^{(1)}$, we obtain the reduced eigenvalue problem
\begin{eqnarray}
\label{first-order-final}
 &\phantom{t}&\mu^2 (1-2\chi_n)^{-1} \left[ (1-\chi_n) c_n - \chi_n c_{1-n} \right] \nonumber  \\
&\phantom{t}& \phantom{text}  =2 \cos(\theta_{n+1}-\theta_n) (c_n - c_{n+1}) + 2 \cos(\theta_{n-1}-\theta_n) (c_n -c_{n-1}),
\end{eqnarray}
where $-N+1 \leq n \leq N$. Recall that the linear system (\ref{first-order-final}) is closed at $2N$ equations
since $\theta_{N+1} = \theta_N + \frac{\pi}{2}$ and $\theta_{-N} = \theta_{-N+1} - \frac{\pi}{2}$.
We can rewrite the reduced eigenvalue problem in the matrix form,
\begin{equation}
\label{matrix-eig-problem}
\mu^2 B(\chi_1,\chi_2,\ldots,\chi_N) {\bf c} = A(\theta_1,\theta_2,\ldots,\theta_N) {\bf c},
\end{equation}
where $A$ and $B$ are $2N \times 2N$ matrices.

The symmetric matrix $A(\theta_1,\theta_2,\ldots,\theta_N)$ coincides with the one considered in \cite{pkf}
for a finite dNLS chain of $2N$ nodes. Although the phase variables $(\theta_1,\theta_2,\ldots,\theta_N)$
depend on $\gamma$, the sign of eigenvalues of $A(\theta_1,\theta_2,\ldots,\theta_N)$ does not depend on $\gamma$,
because of the binary choice for the roots of the sine-function and our definition of
the in-phase and out-of-phase stationary solutions. In particular, all but one eigenvalue
of $A(\theta_1,\theta_2,\ldots,\theta_N)$ are strictly negative for the out-of-phase configuration
with $2 \varphi_1$ and each $\varphi_n$ for $2 \leq n \leq N$ taken in
the interval $\left(\frac{\pi}{2},\frac{3\pi}{2}\right)$. (The last eigenvalue of $A(\theta_1,\theta_2,\ldots,\theta_N)$
is always zero.) More generally, the number of positive eigenvalues of
$A(\theta_1,\theta_2,\ldots,\theta_N)$ coincides with the number of in-phase
differences in the sequence $\theta_{n+1}-\theta_n$ for $-N+2 \leq n \leq N-1$.

Now if $\chi_n < \frac{1}{2}$ for all $n$, the symmetric matrix $B(\chi_1,\chi_2,\ldots,\chi_N)$ is strictly positive definite
because each Gershgorin's circle is bounded from zero in the positive domain. By Sylvester's inertial law theorem
\cite{Lancaster}, the numbers of positive, negative, and zero eigenvalues of $\mu^2$ in the reduced eigenvalue
problem (\ref{matrix-eig-problem}) coincides with those of the matrix $A(\theta_1,\theta_2,\ldots,\theta_N)$.
Therefore, the only spectrally stable stationary solution must have the out-of-phase configuration
for all phase differences in the sequence $\theta_{n+1}-\theta_n$ for $-N+2 \leq n \leq N-1$.

\vspace{0.2cm}

{\bf Case $\chi_n = \frac{1}{2}$ for all $n$:} For $\delta = 0$, the zero eigenvalue
$\Lambda = 0$ of the spectral problem (\ref{spectral-problem-E}) has geometric multiplicity $3N$
with the kernel spanned by the eigenvectors
\begin{equation}
\label{eigenvectors-new}
\left[ \begin{array}{l} \phi_n \\ \psi_n \\ \phi_{1-n} \\ \psi_{1-n} \end{array} \right] = i c_n \left[ \begin{array}{l}
e^{i \theta_n} \\ - e^{-i \theta_n} \\ e^{-i \theta_n} \\ - e^{i \theta_n} \end{array} \right] +
i a_n  \left[ \begin{array}{l}
e^{i \theta_n} \\ 0 \\ -e^{-i \theta_n} \\ 0 \end{array} \right] +
i b_n  \left[ \begin{array}{l}
0 \\ - e^{-i \theta_n} \\ 0 \\ e^{i \theta_n} \end{array} \right], \quad 1 \leq n \leq N,
\end{equation}
Only the first eigenvector in (\ref{eigenvectors-new}) generates a generalized
eigenvector
\begin{equation}
\label{eigenvectors-generalized-new}
\left[ \begin{array}{l} \phi_n \\ \psi_n \\ \phi_{1-n} \\ \psi_{1-n}  \end{array} \right] = \frac{1}{2} c_n \left[ \begin{array}{l}
e^{i \theta_n} \\ e^{-i \theta_n} \\ e^{-i \theta_n} \\ e^{i \theta_n}  \end{array} \right], \quad 1 \leq n \leq N,
\end{equation}
which satisfies a derivative of the spectral problem (\ref{spectral-problem-E}) in $\Lambda$,
hence $N$ generalized eigenvectors exist. Note that the other two eigenvectors in (\ref{eigenvectors-new})
are related to the symmetry of the dNLS chain (\ref{dnls-general}) when $\chi_n = \frac{1}{2}$ for all $n$
and they are preserved in the  spectral problem (\ref{spectral-problem-E}) for $\Omega = 0$ and any $\delta$.
Therefore, we set $a_n = b_n = 0$ in the perturbation expansions below.

For the first-order perturbation theory, we construct the perturbation expansion
$\Lambda = \mu \delta^{1/2} + \mathcal{O}(\delta^{3/2})$ and a linear superposition
of the first eigenvectors in (\ref{eigenvectors-new}) and the generalized eigenvectors (\ref{eigenvectors-generalized-new}).
Proceeding similarly as in the previous case, we obtain the system of difference equations:
\begin{eqnarray*}
& \phantom{t} & -\frac{1}{2} \left( e^{-i \theta_n} \phi_n^{(1)} + e^{i \theta_n} \psi_n^{(1)} +
e^{-i \theta_n} \psi_{1-n}^{(1)} + e^{i \theta_n} \phi_{1-n}^{(1)} \right) =
\frac{i}{2} \mu^2 c_n \\
& \phantom{t} & \phantom{texttexttext} - \gamma_n c_n - i a_n^{(1)} c_n +
i \left( c_{n+1} e^{i (\theta_{n+1}-\theta_n)} + c_{n-1} e^{i (\theta_{n-1}-\theta_n)} \right), \\
& \phantom{t} & -\frac{1}{2} \left( e^{-i \theta_n} \phi_n^{(1)} + e^{i \theta_n} \psi_n^{(1)} +
e^{i \theta_n} \phi_{1-n}^{(1)} + e^{-i \theta_n} \psi_{1-n}^{(1)} \right) =
-\frac{i}{2} \mu^2 c_n \\
& \phantom{t} & \phantom{texttexttext} - \gamma_n c_n + i a_n^{(1)} c_n -
i \left( c_{n+1} e^{-i (\theta_{n+1}-\theta_n)} + c_{n-1} e^{-i (\theta_{n-1}-\theta_n)} \right),
\end{eqnarray*}
which yields the reduced eigenvalue problem
\begin{equation}
\label{first-order-final-new}
\mu^2 c_n = 2 \cos(\theta_{n+1}-\theta_n) (c_n - c_{n+1}) + 2 \cos(\theta_{n-1}-\theta_n) (c_n -c_{n-1}).
\end{equation}
This eigenvalue problem is closed for $1 \leq n \leq N$ subject to the boundary conditions $c_0 = c_1$ and $c_{N+1} = c_N$.
It corresponds to the system of linear equations (\ref{first-order-final}) with $\chi_n = 0$ but
the number of equations is $N$ instead of $2N$. In other words, we obtained the same reduced eigenvalue problem as in \cite{pkf} but
for a finite dNLS chain of $N$ nodes. The phase difference $\theta_1 - \theta_0 = 2 \varphi_1$
becomes irrelevant for stability computations and hence two spectrally stable stationary solutions
exist and have the out-of-phase phase differences in the sequence $\theta_{n+1}-\theta_n$ for $1 \leq n \leq N-1$.
For these two stable stationary solutions, the phase difference $\theta_1 - \theta_0 = 2 \varphi_1$ can be
either in phase or out of phase.
\end{proof}

\begin{theorem}
\label{theorem-stability-2}
Consider the two stationary solutions of Theorem \ref{theorem-nonlocal-bifurcation}  for $M = 1$
under the condition $|\gamma_1| < 1$. If $\chi_1 < \frac{1}{2}$, the in-phase stationary mode is spectrally unstable,
whereas the out-of-phase stationary mode is spectrally stable if and only if the zero equilibrium
for the two disjoint isolated sets $-N+1 \leq n \leq -1$ and $2 \leq n \leq N$ in the dNLS
equation (\ref{dnls-general}) is spectrally stable.
If $\chi_1 = \frac{1}{2}$, both stationary modes are
spectrally stable if and only if the zero equilibrium
for the two disjoint isolated sets $-N+1 \leq n \leq -1$ and $2 \leq n \leq N$ in the dNLS
equation (\ref{dnls-general}) is spectrally stable.
\end{theorem}

\begin{proof}
We still consider the spectral problem (\ref{spectral-problem-E}) but the stationary solution
is now expanded as
\begin{equation}
\label{leading-order-W}
W_n = e^{i \varphi_1} \delta_{n,1} + e^{-i \varphi_1} \delta_{n,0} + \mathcal{O}(\delta),
\end{equation}
where $\delta_{n,j}$ is the Kronecker symbol and $\sin(2 \varphi_1) = \gamma_1$. Because of this expansion,
the spectral problem (\ref{spectral-problem-E}) for $\delta = 0$ has
eigenvalue $\Lambda = 0$ of algebraic multiplicity
$4$ and the pair of semi-simple eigenvalues $\Lambda = \pm i$ of multiplicity
$2N-2$. The double zero eigenvalue splits according to the first-order perturbation theory
described in the proof of Theorem \ref{theorem-stability-1}.

If $\chi_1 \neq \frac{1}{2}$, then only two eigenvectors (\ref{eigenvectors}) exist for $n = 0$ and $n = 1$
and the matrix eigenvalue problem (\ref{first-order-final}) takes the form
\begin{equation}
\label{Mequal1}
\left\{ \begin{array}{l} \mu^2 (1-2\chi_1)^{-1} \left[ (1-\chi_1) c_1 - \chi_1 c_0 \right] = 2 \cos(2 \varphi_1) (c_1 - c_0), \\
\mu^2 (1-2\chi_1)^{-1} \left[ (1-\chi_1) c_0 - \chi_1 c_1 \right] = 2 \cos(2 \varphi_1) (c_0 - c_1). \end{array} \right.
\end{equation}
Therefore, a pair of simple nonzero eigenvalues exists at $\mu^2 = 4 (1 - 2\chi_1) \cos(2 \varphi_1)$.
If $\chi_1 < \frac{1}{2}$, the in-phase
stationary mode with $2 \varphi_1 \in \left(-\frac{\pi}{2},\frac{\pi}{2}\right)$
has a real eigenvalue and is hence unstable right away.
The out-of-phase stationary mode with $2 \varphi_1 \in \left(\frac{\pi}{2},\frac{3 \pi}{2}\right)$
has a pair of imaginary eigenvalues.

If $\chi_1 = \frac{1}{2}$, the zero eigenvalue $\Lambda = 0$ of multiplicity $4$ persists
in the spectral problem (\ref{spectral-problem-E}) for any small $\delta$ with three eigenvectors
(\ref{eigenvectors-new}) and one generalized eigenvector (\ref{eigenvectors-generalized-new}). Therefore,
the zero eigenvalue does not split for $\delta \neq 0$, both for in-phase and out-of-phase stationary modes.

To clarify stability of the stationary modes with respect to the semi-simple eigenvalues $\Lambda = \pm i$,
we again consider the first-order perturbation
theory. For definiteness, we will consider the point $\Lambda = -i$.
Recall (\ref{leading-order-W}), denote
$\Lambda = -i + \mu \delta + \mathcal{O}(\delta^2)$, and use the expansion
\begin{equation}
\label{perturbation-2}
\left[ \begin{array}{l} \phi_n \\ \psi_n \end{array} \right] = c_n \left[ \begin{array}{l}
1 \\ 0 \end{array} \right] + \delta \left[ \begin{array}{l} \phi_n^{(1)} \\ \psi_n^{(1)}  \end{array} \right]
+ \mathcal{O}(\delta^2),
\end{equation}
for $2 \leq n \leq N$ and $-N+1 \leq n \leq -1$, where $\mu$ is a new eigenvalue.
Substituting (\ref{perturbation-2}) to the spectral problem
(\ref{spectral-problem-E}) for these $n$, we obtain the difference equations at the $\mathcal{O}(\delta)$ order:
\begin{eqnarray*}
0 & = & i(\mu + \gamma_n) c_n + c_{n+1} + c_{n-1}, \\
2 \phi_n^{(1)} & = & 0,
\end{eqnarray*}
where the boundary conditions are $c_1 = c_{N+1} = 0$ and $c_{-N} = c_0 = 0$. Because of these
boundary conditions, we actually have two uncoupled eigenvalue problems for
two disjoint sets
$$
S_- = \{ n : \;\; -N+1 \leq n \leq -1 \} \quad
\mbox{\rm and} \quad
S_+ = \{ n : \;\; 2 \leq n \leq N \}.
$$
For $\mu = i \tilde{E}$, we arrive to the same linear eigenvalue problem
\begin{equation}
\label{reduced-eigenvalue}
\tilde{E} c_n = c_{n+1} + c_{n-1} + i \gamma_n c_n, \quad n \in S_{\pm},
\end{equation}
as the one that occurs in the linear dNLS equation for
$S_+$ and $S_-$ subject to the Dirichlet boundary conditions
at the end points. This yields the assertion of the theorem.
\end{proof}

Note that the stability conclusions of Theorem \ref{theorem-stability-1}
change if $\chi_n > \frac{1}{2}$ for at least some $n$.
In this case, the matrix $B(\chi_1,\chi_2,\ldots,\chi_N)$ in the reduced eigenvalue problem (\ref{matrix-eig-problem})
is not positive definite and the signs of eigenvalues $\mu^2$
do not coincide any longer with the signs of eigenvalues of the matrix $A(\theta_1,\theta_2,\ldots,\theta_N)$.
Nevertheless, it is easy to adjust the stability conclusions of Theorem \ref{theorem-stability-2}
for the stationary mode with $M = 1$. It follows from the eigenvalue $\mu^2 = 4 (1 - 2\chi_1) \cos(2 \varphi_1)$
of the reduced eigenvalue problem (\ref{Mequal1}) for $\chi_1 > \frac{1}{2}$ that
the out-phase stationary mode with $2 \varphi_1 \in \left(\frac{\pi}{2},\frac{3 \pi}{2}\right)$
is unstable, whereas the in-phase stationary mode with $2 \varphi_1 \in \left(-\frac{\pi}{2},\frac{\pi}{2}\right)$
is spectrally stable if and only if the zero equilibrium
for the two disjoint isolated sets $-N+1 \leq n \leq -1$ and $2 \leq n \leq N$ in the dNLS
equation (\ref{dnls-general}) is spectrally stable.

For the clustered chain, we have $\gamma_n = \pm \gamma$ for $n \in S_{\pm}$ and
the linear eigenvalue problems (\ref{reduced-eigenvalue}) have $N-1$ pairs of complex
eigenvalues $\tilde{E}$ with ${\rm Im}(\tilde{E}) = \pm \gamma$.
Hence the two stationary modes with $M = 1$ are spectrally unstable in the clustered chain with
any $N \geq 2$.

For the alternating chain, we have $\gamma_n = (-1)^n \gamma$ and the linear eigenvalue problems
(\ref{reduced-eigenvalue}) have $N-1$ pairs of real eigenvalues $\tilde{E}$ if $N$ is odd
and $|\gamma|  < \gamma_{\p\T}$ because the numbers of oscillators with gain and dissipation
are equal to each other. Hence the out-of-phase stationary mode is spectrally stable in the alternating chain with
odd $N$ if $\chi_n < \frac{1}{2}$ for all $n$ and $|\gamma|<\gamma_{\PT}$.
If $N$ is even, however, there exists always one pair of purely imaginary eigenvalues with
${\rm Im}(\tilde{E}) = \pm \gamma$ because the number of oscillators with gain and dissipation
do not match each other by one. Therefore, the two stationary modes with $M = 1$ are spectrally unstable
in the alternating chain with even $N$.

For the defect chain, the defects have to be located at the central sites $n = 0$ and $n = 1$ by
the construction of the $\p\T$-invariant stationary solutions. In this case, all $\gamma_n = 0$,
hence all eigenvalues $\tilde{E}$ of the linear eigenvalue problem (\ref{reduced-eigenvalue}) are
real. Therefore, the out-of-phase stationary mode is spectrally stable in the defect chain if $\chi_n < \frac{1}{2}$ for all $n$.

\begin{figure}
\centering\includegraphics[width=\textwidth]{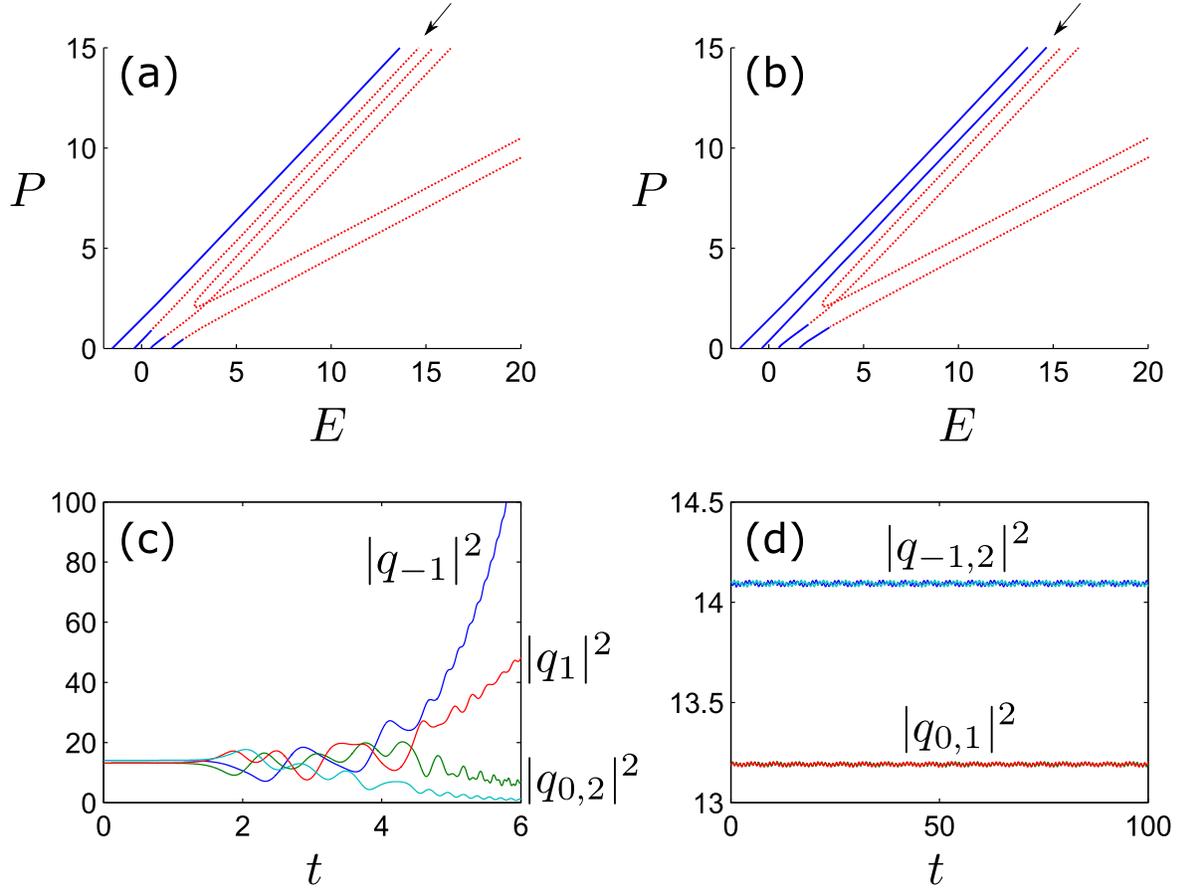}%
\caption{(Color online) Top panels show nonlinear modes of the quadrimer visualized on the plane $P$ vs. $E$.
Stable (unstable) modes are shown with  blue solid  (red dashed) fragments of the curves.
Arrows show the family that is unstable for large $E$ in the left panel  (which corresponds
to $\chi_{1}=\chi_2=0$), but becomes stable for large $E$ in the right panel (which corresponds
to $\chi_{1}=\chi_2=\frac{1}{2}$). Dynamics of a particular stationary mode  belonging to
this family (at $E=13$) is shown in the bottom panels (c) and (d), which also
correspond to $\chi_{1}=\chi_2=0$ and  $\chi_{1}=\chi_2=\frac{1}{2}$,
respectively. Dynamics of the chosen stationary mode is unstable in the
first case and is stable in the second case.  A small  perturbation was added
to the initial data. For all panels $\gamma_1=\frac{1}{2}$ and $\gamma_2 = -\frac{1}{4}$. }
\label{fig}
\end{figure}

\subsection{Numerical illustration}

Let us illustrate the analytical  predictions of Theorems \ref{theorem-general},
\ref{theorem-nonlocal-bifurcation}, \ref{theorem-stability-1}, and \ref{theorem-stability-2} using as a particular example
the quadrimer model ($N=2$). We have numerically identified its stationary modes
of the stationary dNLS equation (\ref{stationary-general}) with $\gamma_1 = \frac{1}{2}$
and $\gamma_2 = -\frac{1}{4}$ for
different $E$ and computed the stability of the stationary modes by means of the direct computation
of the spectrum of the linearized problem (\ref{spectral-problem}).  The results are
summarized in Fig.~\ref{fig}.

In Fig.~\ref{fig}~(a)--(b) for $\chi_{1,2} = 0$ and $\chi_{1,2} = \frac{1}{2}$ respectively,
the computed nonlinear modes are visualized  in the plane $P$ vs. $E$, where the quantity $P$ is defined as
\begin{equation}
P = \frac{1}{2N} \sum_{n=-N+1}^N|w_n|^2,
\end{equation}
and hence is associated with the norm of a stationary solution.
In agreement with Proposition~\ref{theorem-exactly},  the numerical  approach has indicated
exactly six nonlinear modes existing for sufficiently large $E$. These modes
correspond to the six curves going to $E\to \infty$ on Fig.~\ref{fig}~(a)--(b).
Four solutions covered by Theorem~\ref{theorem-general}  correspond to the four upper curves,
while the other two solutions (described by Theorem~\ref{theorem-nonlocal-bifurcation}
with $N=2$ and $M=1$) correspond to  the two lower curves.

The most important difference between panels (a) and (b)  of Fig.~\ref{fig}
(i.e. between the cases with $\chi_{1,2}=0$ and  $\chi_{1,2}=\frac{1}{2}$) stems
from the stability of the modes. In the former case, there exists exactly
one stable solution for sufficiently large $E$, while in the latter case,
there are exactly two stable solutions.
The observed stability features are   in the full agreement
with Theorems~\ref{theorem-stability-1}
and  \ref{theorem-stability-2}. Indeed, Theorem~\ref{theorem-stability-1} ensures that
for large $E$  the system with   $\chi_{1,2}=0$ admits exactly one   stable solution
among the four solutions described by Theorem~\ref{theorem-general}, while the other
two solutions  (which correspond to $N=2$ and $M=1$) are unstable due to
Theorem~\ref{theorem-stability-2}. On the other hand, for    $\chi_{1,2}=\frac{1}{2}$
Theorem~\ref{theorem-stability-1} predicts exactly two stable solutions
among the four solutions described by Theorem~\ref{theorem-general}, while the
other two solutions are unstable by Theorem~\ref{theorem-stability-2}.

We note that all four stationary solutions bifurcating from the four linear modes
of the linear dNLS equation as $P \to 0$ are spectrally stable, according to the
standard local bifurcation theory. Nevertheless, all but one or two modes loss their stability for
larger values of $P$.

The observed difference in the spectral stability was also confirmed by
direct evolutional simulations of the nonlinear mode, for
which the phase differences between sites $-1$ and $0$ and
sites $1$ and $2$ are out-of-phase and the phase difference between sites $0$ and $1$
is in phase. This mode is unstable
for $\chi_{1,2}=0$ but becomes stable for $\chi_{1,2}=\frac{1}{2}$,
see    panels (c) and (d)  of Fig.~\ref{fig}, where the same initial data (subjected to a small perturbation)
display  different evolution depending on  $\chi_{1,2}$.

\section{Conclusion}
\label{sec:concl}

In the present work we analysed the existence and dynamics of solutions for the generalized $\PT$-symmetric
dNLS network consisting of a finite number of sites. The main outcomes of our work   can be summarized as follows.

First, we have revisited the linear case and established the sufficient conditions
for stability or instability of the zero equilibrium.
In particular, we have provided sufficient conditions of the unbroken and broken $\PT$ symmetry
which hold for arbitrary finite number of sites in the network.

Turning to the full nonlinear model and starting with the simplest model of a $\PT$-symmetric dimer,
we have proven the existence of unbounded solutions in a generic case. However, for a specific choice
of the nonlinearity, corresponding to an integrable  model, all solutions stay bounded, provided
the $\PT$-symmetry of the underlying linear dimer is unbroken.

Further we have shown that a finite $\PT$-symmetric network of dNLS equations
possesses stationary solutions. These solutions were analytically constructed by means of the
continuation from the anticontinuum limit (i.e. from the limit of large amplitudes or large
energy or propagation constant). A result of particular importance and novelty is
the classification of all possible stationary modes in the limit of large energies.
More specifically, we have shown that  under certain conditions a $\PT$-symmetric network
consisting of $2N$ sites admits exactly $2^{N+1}-2$ stationary modes (unique up to a gauge transform) in the large-amplitude limit.
We have also described the shape of the stationary modes and found that when approaching the
anticontinuum limit the amplitudes of  the network sites either become large or vanish.
Moreover, the  large-amplitude sites are all grouped together around the center of the network.

Finally, we have examined stability of the found stationary modes counting the number of modes that are stable in the large-amplitude limit.
The obtained analytical results have also been numerically
illustrated for the quadrimer case. The presented numerical results serve as an independent checkup
for the analytical predictions and allow us to show persistence of these stability predictions
far from the large-amplitude limit.

\section*{Acknowledgments}
The work of DEP was supported by the NSERC Discovery grant. The work of DAZ and VVK was
supported by FCT (Portugal) through the grants PTDC/FIS-OPT/1918/2012 and PEst-OE/FIS/UI0618/2011.

\end{document}